\documentclass[12pt]{article}
\usepackage{amsmath,amssymb,amsthm,latexsym,euscript,amscd, authblk}
\usepackage[english]{babel}
\usepackage{mathrsfs}
\usepackage{bm}

\newtheorem*{theorem*}{Theorem}
\newtheorem*{corollary*}{Corollary}
\newtheorem{proposition}{Proposition}
\newtheorem*{remark*}{Remark}

\author{G.G.\,Amosov\thanks{gramos@mi.ras.ru}}

\affil{Steklov Mathematical Institute of Russian Academy of Sciences,
ul. Gubkina 8, Moscow 119991, Russia}

\title{On tomographic representation on the plane of the space of Schwartz operators and its dual}

\begin{document}

\maketitle

\begin{abstract} 
It is shown that the set of optical quantum tomograms can be provided with the topology of Frechet space. In such a case the conjugate space will consist of symbols of quantum observables including all polynomials of the position and momentum operators.
\end{abstract}

\noindent
{ Keywords: Schwartz operator, optical tomogram, dual map}

\section{Introduction} 
Under the experimental homodyne detection \cite{homodyne}, the result of measurement  will be the optical quantum tomogram  $\omega (t,\varphi )$. For each fixed $\varphi \in [0,2\pi)$ it is a probability distribution on the line. The knowledge of optical tomogram for all values of the parameters $(t,\varphi)\in {\mathbb R}\times [0,2\pi)$ allows to reconstruct a quantum state exactly. 
Recently it was shown that \cite {Werner} the set of all density operators with the kernels from the Schwartz space can be equipped the topology of the Frechet space. In this construction the conjugate space will consist of quantum observables including any polynomials of the position and momentum operators. In the present paper we apply such an ideology for the set of optical quantum tomograms.  
Thus, we continue the development of the techniques for the tomographic map and its dual introduced in \cite{Amo1, Amo2}.

\section {Optical tomograms}

Denote $\mathfrak {S}(H)$ the set of quantum states (positive unit trace operators) in the Hilbert space $H$. Let $q$ and $p$ be the standard position and momentum operators for which the Schwartz space  $S({\mathbb R})$ is an essential domain (the domain for all polynomials of $q$ and $p$).
Consider the characteristic function of a state $T \in \mathfrak {S}(H)$
\begin{equation}\label{3}
F_T(x,y)=Tr(T \exp(ixq+iyp)).
\end{equation}
The characteristic function (\ref {3}) determines the set of probability distributions $\omega (t,\varphi )$ on the plane depending
on the parameter $\varphi \in [0,2\pi )$ by the formula
\begin{equation}\label{tomogr}
\omega _T(t,\varphi )=\frac {1}{2\pi }\int \limits _{\mathbb R}e^{-its}F_T(s\cos\varphi,s\sin\varphi)ds.
\end{equation}
The function (\ref {tomogr}) is said to be an optical quantum tomogram.
If $T$ is written in the form of the integral operator in the coordinate representation $(T f)(x)=\int \limits _{\mathbb R}\rho (x,y)f(y)dy$,
we obtain
\begin{equation}\label{charact}
F_T(x,y)=\int \limits _{\mathbb R}e^{ixt}\rho \left (t-\frac {y}{2},t+\frac {y}{2}\right )dt.
\end{equation}
The inverse Fourier transform allows us to reconstruct the characteristic function from its tomogram:
\begin{equation}\label{obr}
F_T(r\cos\varphi ,r\sin\varphi )=\int \limits _{\mathbb R}e^{irt}\omega _T(t,\varphi )dt.
\end{equation}

\section{The map dual to tomographical}

Let us consider the set $\mathcal T$ consisting of linear integral operators  $T$ in the Hilbert space $H=L^2({\mathbb R})$, whose kernels $\rho (\cdot ,\cdot)$ belong to the Schwartz space $S({\mathbb R}^2)$.
For an operator $T\in \mathcal T$ with the kernel $\rho (\cdot ,\cdot)$ one can define the functions (\ref {charact}) and (\ref {tomogr}).
If $T$ is a positive unit trace operator, then  $\omega (X,\varphi )$ is a optical quantum tomogram.
In \cite {Amo1, Amo2} it was introduced the map $T\to f_T(X,\varphi )$, which is dual to the tomographic map$T\to \omega _T(X,\varphi )$ in the sense that
\begin{equation}\label{duality}
\int \limits _0^{2\pi }\int \limits _{\mathbb R}\omega _T(X,\varphi )f_S(X,\varphi )dXd\phi =Tr(TS)
\end{equation}
for all $T,S\in \mathcal T$. Then, in \cite {Amo1, Amo2} it was shown that the map $T\to f_T(X,\varphi )$ can be extended to the class of
operators including polynomials $P(q,p)$ from the position and momentum operators. Such a representation can be also realised in the Cartesian
coordinates \cite{AmoDne}.

In \cite{Amo1} the expression for the dual map in the form of an integral operator was not obtained in the evident form.
It is connected with the representation (\ref {duality}) is inconvinient.
The property $\omega _T(X,\varphi+\pi )=\omega (-X,\varphi)$ was not taken into account. Redefine the dual map claiming that 
\begin{equation}\label{mula2}
\int \limits _0^{2\pi}\int \limits _0^{+\infty}\omega _T(X,\varphi)f_S(X,\varphi)dXd\varphi=Tr(TS).
\end{equation}

\begin{proposition}
Let an integral operator $S\in \mathcal T$ be defined by the kernel $\rho (\cdot ,\cdot)$,
then for $f_S(X,\varphi )$ determined by (\ref {mula2})

\begin{equation}\label{F}
f_S(X,\varphi )=\int \limits _{{\mathbb R}^2}te^{i(X+x\sin\varphi)t}\rho \left (x-\frac {t\sin\varphi}{2},x+\frac {t\sin\varphi}{2}\right )dxdt.
\end{equation}

\end{proposition}

\begin{proof}

Taking into account the form of characteristic function (\ref {charact}) we get
$$
\int \limits _0^{2\pi}\int \limits _0^{+\infty}\omega _T(X,\varphi)f_S(X,\varphi)dXd\varphi =
$$
$$
\frac {1}{2\pi}\int \limits _0^{2\pi}\int \limits _0^{+\infty}\int \limits _{\mathbb R}e^{-iXr}
F_T(r\cos \varphi ,r\sin\varphi)f_S(X,\varphi)dXd\varphi dr\equiv I.
$$
It follows from (\ref {F}) that
$$
\frac {1}{2\pi}\int \limits _{\mathbb R}e^{-iXr}f_S(X,\varphi )dX=rF_S(r\cos\varphi ,r\sin\varphi).
$$
Thus, we obtain
$$
I=\int \limits _0^{2\pi }\int\limits _0^{+\infty}F_T(r\cos\varphi ,r\sin\varphi)rF_S(r\cos\varphi ,r\sin\varphi)drd\varphi=
$$
$$
\int \limits _{{\mathbb R}^2}F_T(x,y)F_S(x,y)dxdy=Tr(TS).
$$
\end{proof}

Let us consider the basis of $H$ consisting of the eigenfunctions of oscillator 
$$
<x|n>=\frac {1}{\pi ^{1/4}\sqrt {2^nn!}}H_n(x)e^{-\frac {x^2}{2}},
$$
where $H_n(x)$ are Hermite polynomials and $n=0,1,2,\dots $. The kernel of rank one operator $|n><m|$ belongs to
the Schwartz space, therefore $|n><m|\in \mathcal T$.

\begin{proposition} The tomographic map transmits  $|n><m|$ to 
$$
\omega_{|n><m|}(X,\varphi )=e^{i(n-m)\varphi}H_n(X)H_m(X)e^{-X^2}
$$
\end{proposition}

\begin{proof}

Notice that
$$
|m><n|=\frac {1}{2}\left ( |m+n><m+n|+i|m+in><m+in|-\right .
$$
$$
\left . (1+i)(|m><m|+|n><n|)\right ).
$$
Due to the linearity of the tomographic map,
$$
\omega _{|m><n|}=\frac {1}{2}\left (\omega _{|m+n><m+n|}+i\omega _{|m+in><m+in|}-(1+i)(\omega _{|m><m|}+\omega _{|n><n|})\right ).
$$
It is known that the tomographic symbol corresponding to the linear combination $|n+\lambda m><n+\lambda m|$ is
$$
\omega (X,\varphi)=\left |<X,n>e^{in\varphi}+\lambda <X,m>e^{im\varphi}\right |^2.
$$
Now it suffices to take $\lambda=1$ and $\lambda =i$. 
\end{proof}

For a fixed $n=0,1,2,\dots $ consider the set of functions $\{H_k(X)H_{k+n}(X)\}_{k=0}^{+\infty }$. Because the functions from this set
are polynomials of decreasing degrees, it the closed linear envelope of the functions from the set there exists the unique system of
functions 
$\{h^{(n)}_k\}_{k=0}^{+\infty }$,
consisting of biorthogonal functions in the sense that
$$
\int \limits _{\mathbb R}e^{-X^2}H_k(X)H_{k+n}(X)h^{(n)}_s(X)dX=\delta _{ks}.
$$

\begin{proposition} The dual map transmits $|n><m|$ to the function 
$$
f_{|n><m|}(X)=h^{(|n-m|)}_{min\{n,m\}}(X)e^{i(n-m)\varphi}.
$$

\end{proposition}

\begin{remark*} In quantum tomography the functions $h^n(X)$ are adopted to call "pattern functions". In \cite {Guta} their properties
are discussed in detail.
\end{remark*}

\begin{proof}

It is straightforward to check that
$$
\int \limits _0^{2\pi}\int \limits _0^{+\infty}\omega _{|k><l|}(X,\varphi )f_{|m><n|}(X,\varphi )dXd\varphi=\delta _{kn}\delta _{lm}.
$$
\end{proof}

\section{The space of Schwartz operators and its dual}

In  \cite{Werner} the notion of a Schwartz operator $T$ in the Hilbert space $H=L^2({\mathbb R})$ was introduced. By the definition the space of Schwartz operators $\mathcal T$
consists of linear operators which are continuous with respect to the set of seminorms 
\begin{equation}\label{normy}
||T||_{n,m,\psi }=||q^np^m T\psi ||_{H},
\end{equation}
where $q$ and $p$ are the position and momentum operators, while $\psi \in S({\mathbb R})$. The conjugate space ${\mathcal T}'$ consisting of linear continuous functionals on ${\mathcal T}$ includes among other things all polynomials $P(q,p)$ from
$q$ and $p$. The duality connection is given by the formula 
$$
T\to Tr(TP(q,p)).
$$
In \cite{Werner} it was shown that $\mathcal T$ is the Frechet space. Moreover, $T\in {\mathcal T}$ if and only if 
it is an integral operator with the kernel $\rho (\cdot ,\cdot)$ belonging to the Schwartz space $S({\mathbb R}^2)$.

\section{Representation of the space of Schwartz operators in the form of optical quantum tomograms.}

Let us extend the domain of integral operator (\ref {charact}) to all kerenels $\rho (\cdot ,\cdot )\in S({\mathbb R}^2)$ (this is the same as a linear extension). On the space $\mathcal T$ consisting of integral operators $(Tf)(x)=\int \limits _{\mathbb R}\rho _T(x,y)f(y)dy,\ f\in H$, define a map $T\to \omega _T(X,\varphi)$ by the formula
\begin{equation}\label {main}
\omega_T(X,\varphi)=\frac {1}{2\pi }\int \limits _{\mathbb R}\int \limits _{\mathbb R}e^{i(s\cos\varphi-X)t}\rho _T\left (s-\frac {t\sin\varphi }{2} ,s+\frac {t\sin\varphi}{2}\right )dsdt.
\end{equation}
The map (\ref {main}) is a composition of (\ref {charact}) and (\ref {tomogr}).

The map (\ref {charact}) is a composition of the affine transformation of the plane and the Fourier transform with respect to one of coordinates.  The map (\ref {tomogr}) is the Fourier transform along a fixed line on the plane. Hence,
the image of the Schwartz space $S({\mathbb R}^2)$ under the map $\rho (\cdot ,\cdot )\to \omega (\cdot ,\varphi)$ belongs to
the Schwartz space $S({\mathbb R})$. Moreover, functions $\omega (X,\cdot )$ are $2\pi $-periodical and infinitely differentiable, i.e.

(i) $\omega (\cdot ,\varphi)\in S({\mathbb R}),\ \varphi \in [0,2\pi )$;

(ii) $\omega (X,\varphi )$ is $2\pi $-periodical and infinitely differentiable.

Nevertheless, not all functions satisfying the properties (i) and (ii) belong to $\hat {\mathcal T}$.

{\bf Example.} {\it Put 
\begin{equation}\label{funkcia}
\omega (t,\varphi )=e^{-t^2}\sin\varphi .
\end{equation}
Function (\ref {funkcia}) satisfies (i) and (ii). Let us find the characteristic function corresponding to  (\ref {funkcia}):
$$
F(r\cos\varphi ,r\sin\varphi )=\int \limits _{\mathbb R}e^{irt}e^{-t^2}\sin\varphi dt=\frac {1}{8\sqrt {2\pi }}e^{-\frac {r^2}{4}}\sin\varphi .
$$ 
It results in
$$
F(x,y)=\frac {1}{8\sqrt {2\pi }}e^{-\frac {x^2+y^2}{2}}\frac {y}{\sqrt {x^2+y^2}}.
$$
One can see that $F(\cdot ,\cdot )\notin S({\mathbb R}^2)$. }

The question arises: what functions belong to $\hat {\mathcal T}$. Taking into account 
$$
t^n=\frac {n!}{2^n}\sum \limits _{m=0}^{[\frac {n}{2}]}\frac {1}{m!(n-2m)!}H_{n-2m}(t)
$$
define a set of functions
$$
g_n(t)=\frac {n!}{2^n}\sum \limits _{m=0}^{[\frac {n}{2}]}\frac {(-i)^{n-2m}}{m!(n-2m)!}H_{n-2m}(t),\ n=0,1,2,3,\dots
$$
Put
$$
f_{m,n}(t,\varphi )=\frac {1}{\sqrt {2\pi}}g_n(t)e^{-t^2}\sin ^m\varphi \cos ^{n-m}\varphi ,\ 0\le m\le n.
$$

\begin{proposition} The following inclusion holds true, $f_{m,n}\in \hat {\mathcal T}$.
\end{proposition}

\begin{proof}

Let us find the characteristic function $F_{m,n}(x,y)$ corresponding to $f_{m,n}$.
Notice that
$$
F_{m,n}(r\cos\varphi,r\sin\varphi )=\int \limits _{\mathbb R}e^{irt}f_{m,n}(t,\varphi)dt=t^ne^{-t^2}\sin ^m\varphi \cos ^{n-m}\varphi.
$$
Hence,
$$
F_{m,n}(x,y)=x^{n-m}y^me^{-\frac {x^2+y^2}{2}}.
$$
The functions $F_{m,n}\in S({\mathbb R}^2)$. Moreover, their linear combinations form a dense set in $S({\mathbb R}^2)$. 

\end{proof}

Below we need the family of functions $\{Q_{m,n}(\varphi ),\ 0\le m\le n\}$ which are biorthogonal to $\{\sin ^m\varphi \cos ^{n-m}\varphi\}$
in its linear envelope \cite {Amo1}. Denote ${\mathcal V}_N$ the linear space spanned by functions $\{f_{m,2N},\ 0\le m\le N\}$. Let ${\mathcal P}_N$ be
the orthogonal projection on the subspace ${\mathcal V}_N$. Define a family of seminorms on the set $\hat {\mathcal T}$ by the formula

\begin{equation}\label{seminorm}
||\omega ||_{m,N}=\left |\int \limits _{0}^{2\pi }\int \limits _{0}^{+\infty}X^{2N}{\mathcal P}_N(\omega (X,\varphi ))Q_{m,2N}(\varphi)d\varphi dX\right |,\ 
\end{equation}
$N=0,1,2,\dots $ 

\begin{theorem*} The map $T\to \omega _T(X,\varphi )$ is continuous with respect to the family of seminorms (\ref {seminorm}).
\end{theorem*}

\begin{proof}

By a definition, $T\to \omega _T(X,\phi )$ determines a map of the space of Schwartz operators  ${\mathcal T}$ to the set $\hat {\mathcal T}$.
Notice that \cite {Amo1}
$$
||\omega _T||_{m,N}=|Tr(P_{2N}(q,p)T_N)|,
$$
where $P_{2N}$ is some polynomial of degree $2N$ from the position and momentum operators $q$ and $p$, while $T_N$ is a restriction of
the density operator to the finite dimensional subspace spanned by $N$ successive excited states of quantum oscillator. 
Hence, the map $T\to \omega _T(X,\varphi)$ is continuous. 

\end{proof}

Theorem implies that the following Corollary holds true.

\begin{corollary*} The map $S\to f_S(X,\phi )$ determines the representation of ${\mathcal T}'$ in the space
$\hat {\mathcal T}'$ equipped with the system of seminorms (\ref {seminorm}).
\end{corollary*}

\section*{Acknowledgments}

The work is supported by the grant of Russian Science Foundation under project No 14-21-00162 and fulfilled in Steklov Mathematical Institute of Russian Academy of Ssciences.

 \end{document}